\documentclass[pt]{article}
\usepackage{a4wide}
\usepackage{amssymb}
\usepackage{xcolor}
\usepackage[normalem]{ulem}
\usepackage{ifthen}
\usepackage{hyperref}

\begin{document}
{\renewcommand{\thefootnote}{\fnsymbol{footnote}}
\begin{center}
{\LARGE  Monopole star products are non-alternative}\\
\vspace{1.5em}
Martin Bojowald,$^1$\footnote{
e-mail address: {\tt
    bojowald@gravity.psu.edu}.}
 Suddhasattwa Brahma,$^{1,2}$\footnote{e-mail address: {\tt suddhasattwa.brahma@gmail.com}} Umut B\"{u}y\"{u}k\c{c}am,$^1$\footnote{e-mail address: {\tt uxb101@psu.edu}} and Thomas Strobl$^3$\footnote{e-mail address: {\tt strobl@math.univ-lyon1.fr}}\\
\vspace{0.5em}
$^1$ Institute for Gravitation and the Cosmos,\\
The Pennsylvania State
University,\\
104 Davey Lab, University Park, PA 16802, USA\\
\vspace{0.5em}
$^2$ Center for Field Theory and Particle Physics,\\
Fudan University,\\
200433 Shanghai, China\\
\vspace{0.5em}
$^3$ Institut Camille Jordan, Universit\'e Claude Bernard Lyon 1,\\
43 boulevard du 11 novembre 1918, 69622 Villeurbanne cedex, France
\vspace{1.5em}
\end{center}
}
\newtheorem{lemma}{Lemma}
\newtheorem{defi}{Definition}
\newtheorem{theo}{Theorem}
\newcommand{\proofend}{\raisebox{1.3mm}{\fbox{\begin{minipage}[b][0cm][b]{0cm}
\end{minipage}}}}
\newenvironment{proof}[1][]{\noindent{\it Proof\ifthenelse{\equal{#1}{}}{}{ (#1)}:} }{\mbox{}\hfill \proofend\\\mbox{}}
\newenvironment{ex}[1][]{\medskip\noindent{\it
    Example\ifthenelse{\equal{#1}{}}{}{ (#1)}:} }{\mbox{}\hfill
  \proofend\\\mbox{}\medskip}

\setcounter{footnote}{0}

\begin{abstract}
  Non-associative algebras appear in some quantum-mechanical systems, for
  instance if a charged particle in a distribution of magnetic monopoles is
  considered. Using methods of deformation quantization it is shown here, that
  algebras for such systems cannot be alternative, i.e.~their associator cannot be completely anti-symmetric.
\end{abstract}


\section{Introduction}

Deformation quantization \cite{DefQuant1,DefQuant2} has been explored much in
the associative setting. If one drops the condition that the star product be
associative, some of the usual methods are no longer available. The
classification of such star products therefore remains open. In this paper, we
present one general result in this direction, motivated by a recent resurgence
of interest in magnetic-monopole systems
\cite{NonGeoNonAss,MSS1,BakasLuest,MSS2,MSS3,WeylStar}, where standard
quantization methods show that associative algebras cannot constitute
consistent quantizations of the relevant observables \cite{JackiwMon,Malcev}.

In the original version of deformation quantization, associativity of the star
product represents an important condition on the coefficients in the formal
power series of the product. If one works with star products without the
condition of associativity, at first sight it may seem easier to find
acceptable versions because they may appear to be subject to fewer consistency
requirements. However, if one is forced to use a non-associative star product
for physical reasons, one is not fully liberated from imposing conditions on
the associator
\begin{equation} \label{DefAsso}
[a,b,c]=a*(b*c) - (a*b)*c\,.
\end{equation}
For a specific set of basic observables, the associator, like the usual
commutator
\begin{equation} \label{DefKomm}
 [a,b] = a*b-b*a\,,
\end{equation}
is prescribed based on physical arguments.

Formulated for position and momentum components as basic observables, the
commutator of an acceptable star product should be
$[q_i,p_j]=i\hbar\{q_i,p_j\}=i\hbar \delta_{ij}$, mimicking the Poisson
bracket. If these are coordinates of a charged particle (with electric charge
$e$) moving in the magnetic field $B^l(q_i)$ of a magnetic monopole
distribution, so that ${\rm div}B=\partial_lB^l\not=0$, the classical brackets
are modified: They are twisted Poisson brackets for which the Jacobi identity
does not hold \cite{TopoTwist,WZWTwist,Twisted}.  An algebra that quantizes
the bracket endows phase-space functions with a new product $\star$ and the
associated commutator (\ref{DefKomm}) and associator (\ref{DefAsso}). The
Jacobiator of the commutator is proportional to the totally antisymmetric part
of the associator and can be non-zero for non-associative $\star$-products.
In the present context, one is led to the relations \cite{JackiwMon,Malcev}
\begin{eqnarray}
 [q_i,q_j]&=& 0 \label{qq}\\
 \left [q_i,p_j\right ]&=& i \hbar \delta_{ij} \\
 \left [p_i,p_j\right]&=&i \hbar  e \epsilon_{ijk} B^k \label{pp}  \\
\left[ q_i,x^I,x^J\right] &=& 0 \\
\left [ p_i,p_j, p_k \right ] &=& -\hbar^2 e \epsilon_{ijk}
 \partial_lB^l \label{Asso}
\end{eqnarray}
to be realized by a star product. Here $(x^I)_{I=1}^6$ is a collective
notation for the Cartesian coordinates $(q_i,p_i)_{i=1}^3$.  In the absence of
a magnetic charge density, one can introduce a canonical momentum $\pi_i$ with
zero brackets for its components. However, the definition, $\pi_i := p_i +
A_i$, makes use of a vector potential $A$ through $B={\rm rot}A$, which does
not exist if ${\rm div}B$ does not vanish. Instead of a zero associator in
standard star products, the specific form of (\ref{Asso}) imposes restrictions
on acceptable star products for magnetic-monopole systems.

Most of the usual properties of quantum mechanics are no longer valid and must
be modified when observables cannot be represented as associative operators on
a Hilbert space. In some studies, a weaker condition given by an alternative
algebra has been found advantageous \cite{OctQM,NonAss,NonAssEffPot}---if it
can be realized. An alternative algebra is one where the associator
(\ref{DefAsso}) is completely antisymmetric, or, equivalently, where the
$*$-product obeys
\begin{eqnarray}
a*(a*b)&=&(a*a)*b\nonumber\\
(a*b)*b&=&a*(b*b)
\end{eqnarray}
for any $a,b$ in the algebra. Many well-known non-associative algebras are of
this form, such as the octonionic ones. Requiring an algebra to be
alternative, provides a priori a tempting option for the case of a charged
particle in the background of magnetic monopoles, in particular in view of the
total anti-symmetry of the basic relation (\ref{Asso}).

However, in this report we demonstrate the impossibility of such an algebra as
a set of quantized observables of a charged particle in the presence of
magnetic monopole densities, obtained by deformation quantization.  While
(\ref{Asso}) implies a totally antisymmetric associator for linear functions
of the basic observables, the associator of general algebra elements is not
guaranteed to be totally antisymmetric.  Different examples for algebras
consistent with the relations (\ref{qq})--(\ref{Asso}) have been constructed
using star products \cite{NonGeoNonAss,MSS1,BakasLuest,MSS2,MSS3,WeylStar},
one of which has explicitly been shown to be non-alternative \cite{Schupp}. In
what follows, we will analyze the possibility of alternative monopole star
products in general terms, using deformation theory, the basics of which we
first recall in the next section.

\section{Deformation quantization with non-associativity}

The classical theory is described by the commutative algebra of smooth
functions on $T^*\mathbb{R}^3$, equipped with the bivector field\footnote{We
  set the electric charge to $e=1$ from now on.}
\begin{equation} \label{Pi}
\Pi = \left(\frac{\partial}{\partial q_i} + \epsilon_{jik}B^k(q)
  \frac{\partial}{\partial p_j}\right) \wedge \frac{\partial}{\partial p_i}
\, ,
\end{equation}
in the canonical linear coordinates $(x^I)_{I=1}^6\equiv
(q_1,q_2,q_3,p_1,p_2,p_3)$. For a vector field $B$ with non-vanishing
divergence, this is only a twisted Poisson bivector: Its Schouten bracket with
itself does not vanish but is given by
\begin{equation}\label{PiPi}
\frac{1}{2}[\Pi,\Pi] = \Pi^\sharp(H)
\end{equation}
where the 3-form $H$ takes the form
\begin{equation}
H= \pi^* \mathrm{d} B \, .
\end{equation}
Here the magnetic field $B$ is considered a 2-form on $\mathbb{R}^3$ by means
of $B=\epsilon_{ijk} B^i \mathrm{d} q^j \wedge \mathrm{d} q^k$ and $\pi \colon
T^*\mathbb{R}^3 \to \mathbb{R}^3$ is the canonical projection. Maxwell's
equations link $\mathrm{d}B$ directly to the magnetic monopole density:
$\mathrm{d}B= \ast \rho_{\rm magnetic}$.

The bivector field $\Pi$ then induces the following bracket on the functions
$f,g \in C^\infty(T^*\mathbb{R}^3)$,
\begin{equation}
\{ f,g\} = \frac{1}{2} \Pi^{IJ}(x) \frac{\partial f}{\partial x^I}
\frac{\partial g}{\partial x^J} \, .
\end{equation}
This bracket is an antisymmetric bi-derivation, but no longer a Lie bracket
and thus not a Poisson bracket: the r.h.s.~of (\ref{PiPi}) provides
precisely the non-zero Jacobiator.

\subsection{Star product}

Deformation quantization turns the classical commutative algebra
$(C^\infty(T^*\mathbb{R}^3),\cdot)$ into the quantum algebra ${\cal A} :=
(C^\infty(T^*\mathbb{R}^3)[[\lambda]],\star)$, where $\lambda=\frac{1}{2}i\hbar$
is considered as a formal deformation or expansion parameter:
\begin{eqnarray} \label{star}
 f\star g = \sum_{j=0}^{\infty} \lambda^j B_j(f,g)\,.
\end{eqnarray}
Here $B_j\colon {\cal A}\times{\cal A}\to {\mathbb C}$ are bilinear maps on
$\mathcal{A}$.\footnote{Using the same letter for these bilinear maps and the
  magnetic field should not cause confusion.} To zeroth order in $\lambda$, we
have the classical product given by pointwise multiplication, $B_0(f,g)=f\cdot
g \equiv fg$. Following \cite{DefQuant1}, we will assume that $B_j$ is a
bi-differential operator of maximum degree $j$ which is zero on constants for
strictly positive $j$:
\begin{eqnarray}
B_j(f,g) &=& \sum_{k,l =1}^j B_j^{k,l}(f,g) \qquad \mathrm{for} \quad j\geq 1\\
 B_j^{k,l}(f,g) &=& \sum_{I_1, \ldots, I_k,J_1,\ldots, J_l=1}^6 B_{j;I_1,
   \ldots, I_k,J_1,\ldots, J_l}^{k,l}(q)
\frac{\partial^k f}{\partial x^{I_1} \cdots \partial x^{I_k}} \frac{\partial^l
  g}{\partial x^{J_1} \cdots \partial x^{J_l}}
\end{eqnarray}
The property implies in particular that the star product defines a unital
algebra, with the unit function as unit.

Let us for a moment assume that $\star$ would be associative. In this case, we
would have that the commutator (\ref{DefKomm}) evidently satisfies the Jacobi
identity and also that $[f, g \star h] = [f,g] \star h + g \star [f,h]$. Both
equations together, evaluated at lowest non-vanishing order in $\lambda$, imply
that the antisymmetric part $B^-_1(f,g)=\frac{1}{2}(B_1(f,g)-B_1(g,f))$ of
$B_1(f,g)$ is a Poisson bivector. On the other hand, for physical reasons, we
want that the antisymmetric part of the first order deformation is determined
by the classical bracket:
\begin{equation}\label{B1-}
 B_1^-(f,g)=\{f,g\} \, .
\end{equation}
This then shows that the $\star$-product cannot be associative for the
deformation quantization of the above classical system, cf., in particular,
Eq.~(\ref{PiPi})---as anticipated already in the Introduction.

In fact, in the present article, we want to strengthen eq.~(\ref{B1-}) in a
two-fold way: First, we require in addition that $B_1$ is antisymmetric itself
already, so that
\begin{equation} \label{B1}
 B_1(f,g)=\{f,g\}\, .
\end{equation}
This, in fact, is not really a restriction: it can be shown that every star
product either satisfies this condition or has an equivalent deformation for
which (\ref{B1}) is fulfilled. We will come back to this below and assume it
for now in any case. Second, we want that for linear coordinate functions on
$T^*\mathbb{R}^3$ the bracket determines the commutator even to
next-to-leading order, i.e.~we require
\begin{equation} \label{Dist}
 \frac{x^I \star x^J -x^J \star x^I}{i\hbar}=\{x^I,x^J\}+O(\hbar^2) \, .
\end{equation}
The first condition is equivalent to requiring $B^+_1(f,g)=0$ for all
functions $f,g$, the second one to demanding
\begin{equation} \label{B2-x}
B_2^-(x^I,x^J)=0 \, .
\end{equation}
We remark in parenthesis that the equation (\ref{Dist}) is implied if the
$x^I$ are implemented as distinguished observables in the sense of
\cite{DefQuant2}.

\subsection{Monopole star products}

Since we found above that the associator of the monopole star product cannot
be zero, we also expand it into a formal power series in $\lambda$:
\begin{eqnarray} \label{Afgh}
 A(f, g, h)= f\star(g\star h) - (f\star g)\star h := \sum_{j=0}^\infty \lambda^j
 A_j(f, g, h)\,.
\end{eqnarray}
The maps $B_i$ and $A_j$ are not independent; in fact, $A_j$ is determined by
the $B_i$ with $i\leq j$. It is easy to evaluate the low orders: We always
have $A_0=0$, because the point-wise multiplication of phase-space functions is
associative. At first order, we have
\begin{equation} \label{A1}
  A_1(f,g,h) = fB_1(g,h)-B_1(f,g)h+B_1(f,gh)-B_1(fg,h)=0
\end{equation}
simply since $B_1$ is bi-differential of order $(1,1)$.

At second order, one finds
\begin{eqnarray} \label{A2B}
A_2(f,g,h) &=&
  fB_2(g,h)-B_2(f,g)h+B_2(f,gh)-B_2(fg,h)\nonumber\\
&& + B_1(f,B_1(g,h))- B_1(B_1(f,g),h)\,.
\end{eqnarray}
For a non-associative star product, the coefficient $A_2$, as the first
non-zero one in the expansion (\ref{Afgh}), plays a role similar to the
coefficient $B_1$ in specifying conditions on the star product as a
quantization of the classical bracket. The totally antisymmetric contribution
\begin{eqnarray} \label{anti} A_2^-(f, g, h)&:=&\frac{1}{6}\left(A_2(f, g,
    h)+A_2(h, f,
    g)+A_2(g, h, f)\right.\nonumber\\
  & &\left.\,\,\,\,-A_2(f, h, g)-A_2(g, f, h)-A_2(h, g, f)\right)\nonumber
\end{eqnarray}
to $A_2$, in view of (\ref{B2-x}), only depends on $B_1$ \emph{if} it is
evaluated on linear functions of the basic variables $x^I$: We have
\begin{equation}
    A_2^-(x^I,x^J,x^K)= \frac{1}{2}J(x^I,x^J,x^K)
\end{equation}
where $J(f,g,h)$ is the Jacobiator of $B_1$, i.e.~of the classical bracket $\{
\cdot, \cdot \}$.  In particular, $A_2^-(p_1,p_2,p_3)= 4\pi^*\mathrm{d} B$ for
a star product that quantizes a twisted Poisson bivector obeying
(\ref{PiPi}). It is then consistent to assume that
$A_2(p_1,p_2,p_3)=A_2^-(p_1,p_2,p_3)$ is totally antisymmetric, as written in
the basic relation (\ref{Asso}). The basic relations do not give us direct
statements about $A_2$ evaluated on functions not linear in the global
coordinates $x^I$. We will assume that $A_2(f,g,h)$ can be chosen totally
antisymmetric even in this case --- since our aim is to prove that monopole
star products cannot be alternative, there would be nothing to show if this
assumption were violated.  However, this condition does not already imply that
the star product is alternative, since non-linear functions generically lead
to contributions to $A(f,g,h)$ of higher order in $\lambda$, which do not
directly follow from simple combinations of the basic relations (\ref{Asso}).

We summarize our conditions on $A_2$ in
\begin{defi}\label{Defmono}
  A {\em monopole star product} is a non-associative star product $\star$ on
  $C^\infty(T^*\mathbb{R}^3)[[\lambda]]$ such that (\ref{Dist}) holds, its
  associator to second order in $\lambda$ is totally antisymmetric and further
  obeys the following conditions:
 \begin{enumerate}
 \item $A_2(p_1,p_2,p_3)\not=0$,
 \item $A_2(q_i,x^I,x^J)=0$ for all $i=1,2,3$ and $I,J=1,\ldots,6$, and
 \item $B_1(q_i,A_2(p_1,p_2,p_3))=0$ for $i=1,2,3$.
\end{enumerate}
where $(x^I)_{I=1}^6=(q_1,q_2,q_2,p_1,p_2,p_3)$ are the canonical linear
coordinates on $T^*\mathbb{R}^3$.
\end{defi}

\subsection{Hochschild cohomology}

For an associative algebra ${\cal A}$, the space of multilinear maps from
${\cal A}$ to itself can be equipped with a coboundary operator ${\rm d}$,
used in Hochschild cohomology. For a multilinear map $\phi\colon {\cal
  A}^{\otimes n}\to {\cal A}$ of $n$ arguments, ${\rm d}\phi$ is a multilinear
function of $n+1$ arguments given by
\begin{eqnarray}
 {\rm d}\phi(a_0,a_1,\ldots,a_n)&=& a_0\cdot \phi(a_1,\ldots,a_n)+ \sum_{j=0}^{n-1}
 (-1)^j \phi(a_0,\ldots, a_{j-1}, a_j\cdot a_{j+1},a_{j+2},\ldots,a_n)\nonumber\\
 &&+ (-1)^{n}
 \phi(a_0,\ldots,a_{n-1})\cdot a_n\,. \label{dphi}
\end{eqnarray}

Hochschild cohomology plays an important role in classifying equivalent star
products with respect to a redefinition of higher orders in a
$\lambda$-expansion: If
\begin{equation}
 D(f)=\sum_{j=0}^{\infty} D_j(f)\lambda^j
\end{equation}
with linear differential operators $D_j$ starting with $D_0 = \mathrm{id}$,
for any given star product $\star$ a new product $\star'$ can be defined by
means of
\begin{equation}
 D(f)\star' D(g) = D(f\star g)\,.
\end{equation}
The condition on $D_0$ ensures that $D$ is invertible as a map on formal power
series. If functions in $C^{\infty}(M)$ are written as symbols of operators,
for instance by a Weyl correspondence, a non-trivial map $D$ changes the
factor-ordering choice in the correspondence.  To first order, $B_1'=B_1-{\rm
  d}D_1$ while ${\rm d}B_1=0$; see (\ref{A1}). The first Hochschild cohomology
therefore classifies inequivalent choices of $B_1$ which cannot be related by
a different choice of factor ordering. For a given bracket $\{\cdot,\cdot\}$,
all star products quantizing it respect the condition (\ref{B1-}), but not
necessarily (\ref{B1}).

If ${\cal A}$ is not associative, $\not\!{\rm d}$, defined just like ${\rm d}$
for an associative algebra, is not a coboundary operator: For a linear
function $\phi\colon {\cal A}\to {\cal A}$, we have
\begin{equation}
 \not\!{\rm d}\phi(a_0,a_1) =
 a_0\star \phi(a_1)-\phi(a_0\star a_1)+ \phi(a_0)\star a_1
\end{equation}
and
\begin{equation}
  {\not\!{\rm d}}^2\phi(a_0,a_1,a_2) = A(a_0,a_1,\phi(a_2))+
  A(a_0,\phi_0(a_1),a_2)+
  A(\phi(a_0),a_1,a_2)- \phi_0(A(a_0,a_1,a_2))
\end{equation}
with the associator $A$. Therefore, Hochschild cohomology is not available for
non-associative algebras. However, the coboundary operator ${\rm d}$ of the
classical associative commutative algebra of smooth functions may still be
used in constructing non-associative deformations, as we will do below. For
instance, the product in (\ref{A1}) refers to $\cdot$, not to
$\star$. Moreover, we can refer to the standard argument \cite{DS} for
changing the star product within its equivalence class to show that the
symmetric part in $B_1$ can always be set to zero and (\ref{B1}) be
achieved. Thus, up to operator ordering, we can always assume that $B_1$ is
given by the classical bracket, even if it is not Poisson, but for example
twisted Poisson as here.

\section{The main result}

Our main result is
\begin{theo} \label{Theorem}
Let $\star$ be a monopole star product as defined
  above, cf.~Definition \ref{Defmono}. Then the associator $A(f,g,h)\equiv
  f\star(g\star h)-(f\star g)\star h$ cannot be totally antisymmetric in its
  arguments.
\end{theo}

We will prove this result by making use of three lemmas:
\begin{lemma}\label{lemma1}
  Let $\star$ be a star product obeying (\ref{Dist}).  If $\star$ is flexible
  at second order, that is $A_2(f,g,h)=-A_2(h,g,f)$, then $B_2$ is symmetric.
\end{lemma}

\begin{proof}
 We evaluate $A_2$ in (\ref{A2B}) on functions with $f=h$, writing the result as
\begin{eqnarray}
 A_2(f,g,f)&=&fB_2(g,f)-B_2(f,g)f+B_2(f,gf)-B_2(fg,f) \nonumber\\
 & =& -2fB_2^-(f,g)
 +2B_2^-(f,fg)
\end{eqnarray}
using the antisymmetric part $B_2^-(f,g):=\frac{1}{2}(B_2(f,g)-B_2(g,f))$ of
$B_2$. If $A_2(f,g,h)=-A_2(h,g,f)$ holds, $A_2(f,g,f)=0$, and we obtain
\begin{equation}
 B_2^-(f,fg)=fB_2^-(f,g)\,.
\end{equation}
For an antisymmetric bi-differential form, this equation can hold only if the
degree is $(1,1)$.  However, if $B_2^-$ has a contribution of degree $(1,1)$,
(\ref{B2-x}) cannot hold.  Therefore, $B_2^-=0$ and $B_2$ is symmetric.
\end{proof}

In particular, the conclusion holds for a monopole star product (\ref{star}).
All explicit star products that have been constructed for monopole systems
indeed have a symmetric $B_2$. For associative star products, Kontsevich's
formula \cite{Kontsevich} has the same property. If symmetry of $B_j$ holds at
all even orders $j$, the star product gives rise to a formal deformation of
the twisted Poisson bracket by powers of $\lambda^2$, or a Vey deformation as
defined in \cite{DefQuant1}.

\begin{lemma} \label{lemma2}
  If (\ref{star}) is a star product with symmetric $B_2$, then the totally
  anti-symmetric part of $A_3$ is equal to zero.
\end{lemma}

\begin{proof}
Using the definition of the associator and the star product, we derive
\begin{eqnarray}\label{A3}
A_3(f, g, h) &=& {\rm d}B_3(f, g, h) + B_2(f, B_1(g,h))\nonumber\\ &
&\hspace{3mm}-B_2(B_1(f,g), h) + B_1(f, B_2(g,h)) - B_1(B_2(f,g), h)\,,
\end{eqnarray}
where $\mathrm{d}$ is the coboundary operator of Hochschild cohomology,
cf.~eq.~(\ref{dphi}). In particular, ${\rm d}B_3(f,g,h) \equiv f B_3(g, h)
+B_3(f, gh) -hB_3(f,g) - B_3(fg, h)$. The totally anti-symmetric part $A_3^-$
of $A_3$, defined as in (\ref{anti}), is given by
\begin{eqnarray}\label{A3final}
3A_3^-(f, g, h)&=&B_2^-(f, 2B_1^-(g,h))+ B_2^-(h, 2B_1^-(f,g)) +
B_2^-(g, 2B_1^-(h,f)) \\
&& +B_1^-(f, 2B_2^-(g,h)) +B_1^-(f, 2B_2^-(g,h)) +B_1^-(f,
2B_2^-(g,h))\nonumber
\end{eqnarray}
where, as before, $B_j^-(f,g)=\frac{1}{2}\left(B_j(f,g) - B_j(g,f)\right)$ is
the antisymmetric part of $B_j$.\footnote{See App.~\ref{a:A3} for a detailed
derivation of (\ref{A3final}).} Since all terms on the right-hand side of
(\ref{A3final}) contain a $B_2^-$,
$B_2^-=0$ implies $A_3^-=0$.
\end{proof}
We remark that for the last conclusion it is important that the antisymmetric
part of $A_3$, unlike the full $A_3$, does not depend on $B_3$.

\begin{lemma} \label{lemma3}
 Let $\star$ be a star product such that
\begin{eqnarray}
O(f,g,h,k) &:=& A_2(f,g,B_1(h,k)) -A_2(f,B_1(g,h),k)
+A_2(B_1(f,g),h,k)\nonumber\\
&&+B_1(A_2(g,h,k),f) -B_1(A_2(f,g,h),k)
\end{eqnarray}
is not identically zero. Then the third-order contribution $A_3$ to the
associator is non-zero.
\end{lemma}

\begin{proof}
 Again, we use the Hochschild coboundary operator and
consider
\begin{eqnarray}
{\rm d}A_3(f,g,h,k)= fA_3(g,h,k) -A_3(fg,h,k) +A_3(f,gh,k)
  -A_3(f,g,hk) +kA_3(f,g,h)\,.
\end{eqnarray}
Our goal is to show that ${\rm d}A_3$ is non-zero for algebras with non-zero
$O$, which implies immediately also that $A_3\neq 0$. The Pentagon
identity
\begin{eqnarray}\label{Pentagon}
f\star A(g,h,k) +A(f,g,h)\star k=A(f\star g,h,k) -A(f,g\star h,k)
+A(f,g,h\star k)
\end{eqnarray}
for non-associative algebras can be used for a compact proof of this
statement.  Expanding it to third order in $\lambda$, we obtain
\begin{eqnarray}\label{Pentagonh3}
& &f A_3(g,h,k) +B_1(f, A_2(g,h,k)) +k A_3(f,g,h) +B_1(A_2(f,g,h), k)\nonumber\\
&=&A_3(fg,h,k) -A_3(f,gh,k) +A_3(f,g,hk) \nonumber\\
& &+A_2(B_1(f,g),h,k) -A_2(f,B_1(g,h),k) +A_2(f,g,B_1(h,k))
\end{eqnarray}
where we used $A_1=0$, cf.~eq.~(\ref{A1}).
These terms can be organized to obtain
\begin{eqnarray}\label{dA3}
{\rm d}A_3(f,g,h,k)&=&A_2(f,g,B_1(h,k)) -A_2(f,B_1(g,h),k)
+A_2(B_1(f,g),h,k)\nonumber\\
& &+B_1(A_2(g,h,k),f) -B_1(A_2(f,g,h),k)\,.
\end{eqnarray}
Alternatively, one can prove directly that ${\rm d}A_3$ is of this form
without invoking the Pentagon identity, as shown in App.~\ref{a:Pentagon}.
The right-hand side of this equation is equal to $O(f,g,h,k)$. If it is not
identically zero, $A_3$ is non-zero.
\end{proof}

We are now ready to prove our main result:
$\vspace{1.5mm}$


\noindent {\bf Proof} (of Theorem~\ref{Theorem}){\bf :}
  By Lemmas~\ref{lemma1} and \ref{lemma2}, a monopole star product has an
  $A_3$ with zero totally antisymmetric part. If the star product is
  alternative, we must then have $A_3=0$.  If the obstruction $O$ provided by
  Lemma~\ref{lemma3} is not identically zero, however, it is not possible that
  $A_3=0$. We now show that $O\not=0$ for a monopole star product, discussing
  two cases separately depending on whether the associator (the monopole
  density) is constant or a function of the position.

  For a constant associator, we may choose $f=p_1$, $g=p_2$, $h=p_3$ and
  $k=q_3p_3$. Using the twisted Poisson bracket for $B_1$, all but the first
  term in $O(f,g,h,k)$ are zero, while $A_2(f,g,B_1(h,k))$ is proportional to
  the monopole density and therefore non-zero.

If the monopole density is not constant, we specialize $O(f,g,h,k)$ to
\begin{equation}
 O(f,g,h,g) = A_2(B_1(f,g),h,g)- B_1(A_2(f,g,h),g)\,.
\end{equation}
Since the associator is not constant, it depends on at least one position
coordinate, say $q_1$ without loss of generality.
If we then choose $f=p_2$, $g=p_1$ and $h=p_3$ we have
$B_1(A_2(f,g,h),g)\not=0$ while $A_2(B_1(f,g),h,g)=0$. $\hfill \square$
$\vspace{1.5mm}$

The conclusion is independent of the choice of the star product within an
equivalence class, with \cite{MSS1} or \cite{WeylStar} as concrete examples,
because alternativity is independent of the choice of the ordering (the
``gauge'') \cite{WeylStarAlt2}.

More generally, Lemma~\ref{lemma3} gives us an obstruction to alternativity
which only depends on $B_1$ and $A_2$, and therefore can be tested for general
non-associative star products more easily than the full associator.

\section{Monopole Weyl star product}

Two different star products have been proposed recently for the
magnetic-monopole system, one by using the Kontsevich formula
\cite{NonGeoNonAss,MSS1,BakasLuest,MSS2,MSS3}, and one from Weyl products
\cite{WeylStar}. The former is known to be non-alternative \cite{Schupp}.
Since it satisfies our assumptions, it provides an explicit example for our
general result. We now discuss the star product of \cite{WeylStar} in more
detail.

\begin{ex}[Weyl star product]
  The star product of \cite{WeylStar} has the first coefficient
  $B_1(f,g)=\frac{1}{2}\{f,g\}$ with an atisymmetric bracket
  $\{f,g\}=\frac{1}{2}\Pi^{IJ}\partial_I f\partial_J g$ given by an arbitrary
  bivector $\Pi^{IJ}$. It can therefore be applied to monopole star
  products. The second coefficient is
\begin{equation}
 B_2(f,g)=-\frac{1}{2}\Pi^{IJ}\Pi^{KL} (\partial_I\partial_K f)
 (\partial_J\partial_L g) - \frac{1}{3}
 \Pi^{IJ}\partial_J\Pi^{KL}\left((\partial_I\partial_Kf)
   (\partial_Lg)-(\partial_Kf)(\partial_I\partial_L g)\right)\,,
\end{equation}
transferred to our notation. It obeys our assumptions. In particular, $B_2$
has no contribution of bi-differential degree $(1,1)$, and it is symmetric
thanks to the antisymmetry of the twisted Poisson tensor $\Pi^{IJ}$.
Therefore, our conditions on monopole star products are satisfied and the
algebra cannot be alternative.\footnote{This star product has been conjectured
  to be alternative in \cite{WeylStarAlt}, with a proof suggested in
  \cite{WeylStarAlt2}. However, the arguments given are not complete: They are
  based on a computation of the associator
  $A_{\xi,\eta,\zeta}:=A(\exp(i\xi\cdot z),\exp(i\eta\cdot z),\exp(i\zeta\cdot
  z))$ with phase-space variables $z$, together with a Fourier representation
  $f(z)=\int {\rm d}\mu f(\xi)\exp(i\xi\cdot z)$ of smooth functions. The
  direct calculation of $A_{\xi,\eta,\zeta}$ shows that it is zero whenever
  two of its arguments are equal. If $A_{\xi,\eta,\zeta}$ were tri-linear in
  $(\xi,\eta,\zeta)$, this fact would imply that it is antisymmetric, which
  would imply alternativity. However, $A_{\xi,\eta,\zeta}$ is not tri-linear
  in $(\xi,\eta,\zeta)$ but rather in $(\exp(i\xi\cdot z),\exp(i\eta\cdot
  z),\exp(i\zeta\cdot z))$, and antisymmetry is not implied. In fact, direct
  inspection of the result given in \cite{WeylStarAlt2} shows that
  $A_{\xi,\eta,\zeta}$ is not antisymmetric in $(\xi,\eta,\zeta)$, even though
  it is zero whenever two of its arguments are equal.}
\end{ex}

In \cite{WeylStarAlt,WeylStarAlt2}, an explicit expression for $B_3$ is given
as well. It is therefore possible to compute $A_3$ in specific examples and
show that it is not totally antisymmetric. In particular, for monopole star
products, it is not difficult to find functions $f(p_1,p_2,p_3)$ such that
$A_3(f,f,f)\not=0$.

\begin{lemma}
  Let $\star$ be a Weyl star product on $C^{\infty}(T^*{\mathbb
    R}^3)[[\lambda]]$ according to \cite{WeylStar} which quantizes a twisted
  Poisson tensor (\ref{Pi}), and let $f(p_1,p_2,p_3)$ be a function of the
  fiber coordinates of $T^*{\mathbb R}^3$ such that
  $\partial_{p_i}\partial_{p_j}f=0$ whenever $i\not= j$. The third coefficient
  of the associator of $\star$ then obeys
\begin{equation} \label{A3f}
 A_3(f,f,f) =
 \frac{4}{3}i\left(\partial_{q_1}\Pi^{p_2p_3}+ \partial_{q_2}\Pi^{p_3p_1}+
\partial_{q_3}\Pi^{p_1p_2}\right)  \sum_{\sigma\in
  Z_3}\Pi^{p_{\sigma(1)}p_{\sigma(2)}} \partial_{p_{\sigma(3)}}f
\partial_{p_{\sigma(1)}}^2f \partial_{p_{\sigma(2)}}^2f \,,
\end{equation}
summing over elements of the alternating group $A_3=Z_3$ of cyclic
permutations.
\end{lemma}
\begin{proof}
  We have explicitly computed $A_3(f,f,f)$ for arbitrary $f$ using Cadabra
  software \cite{Cadabra1,Cadabra2}:
\begin{eqnarray}
A_3(f,f,f) &=& \frac{2i}{3}\Bigl( {\Pi}^{LM} {\partial}_{L}{{\Pi}^{NO}}\,
{\partial}_{N}{{\Pi}^{PQ}}\,  {\partial}_{M}{f}\,  {\partial}_{P}{f}\,
{\partial}_{O}\partial_{Q}{f}\,\nonumber\\
 &&  \quad-  {\Pi}^{LM}
{\partial}_{L}{{\Pi}^{NO}}\,
{\partial}_{N}{{\Pi}^{PQ}}\,  {\partial}_{O}{f}\,  {\partial}_{P}{f}\,
{\partial}_{M}\partial_{Q}{f}\,\nonumber  \\
&& \quad-2\; {\Pi}^{LM} {\Pi}^{NO} {\partial}_{L}{{\Pi}^{PQ}}\,
{\partial}_{P}{f}\,  {\partial}_M\partial_N{f}\,
{\partial}_{O}\partial_{Q}{f}\,\nonumber \\
&& \quad+  {\Pi}^{LM} {\Pi}^{NO} {\partial}_{L}{{\Pi}^{PQ}}\,
{\partial}_{M}{f}\,  {\partial}_{NP}{f}\,
{\partial}_{O}\partial_{Q}{f}\Bigr)\,. \label{A3fff}
\end{eqnarray}
For a monopole star product, the bivector $\Pi$ is a function only of the
position coordinates $q_i$ via the magnetic field. Therefore, $L$ and $N$ must
be position indices for non-zero contributions in the first two terms of
(\ref{A3fff}). These terms are then identically zero because each contains a
factor of $\partial_L\Pi^{NO}$, which is zero for a bivector of the form
(\ref{Pi}).

In the third term, only $L$ is required to be a position index, while $M$,
$N$, $O$, $P$, and $Q$ are momentum indices if $f$ depends only on
momenta. The components $\Pi^{LM}$ then equal $\delta^{LM}$
since they contain one position and one momentum index. The remaining terms in
(\ref{A3fff}) yield
\begin{eqnarray*}
\frac{3}{2i}A_3(f,f,f) &= &-2\, \Pi^{NO}({\partial}_{q_1}{{\Pi}^{UQ}}\,
{\partial}_{U}{f}\,{\partial}_{p_1}\partial_{N}{f}\,{\partial}_{O}
\partial_{Q}{f}\,  +{\partial}_{q_2}{{\Pi}^{UQ}}\,
{\partial}_{U}{f}\,{\partial}_{p_2}\partial_{N}{f}\,{\partial}_{O}
\partial_{Q}{f}\,\\
&& \qquad\quad   +{\partial}_{q_3}{{\Pi}^{UQ}}\,
{\partial}_{U}{f}\,{\partial}_{p_3}\partial_{N}{f}\,
{\partial}_{O}\partial_{Q}{f}\,) \nonumber \\
&&+{\Pi}^{NO} ({\partial}_{q_1}{{\Pi}^{UQ}}\,
{\partial}_{p_1}{f}\,  {\partial}_{N}\partial_{U}{f}\,
{\partial}_{O}\partial_{Q}{f} +{\partial}_{q_2}{{\Pi}^{UQ}}\,
{\partial}_{p_2}{f}\,  {\partial}_{N}\partial_{U}{f}\,
{\partial}_{O}\partial_{Q}{f}\\
&&\qquad\quad + {\partial}_{q_3}{{\Pi}^{UQ}}\,
{\partial}_{p_3}{f}\,  {\partial}_{N}\partial_{U}{f}\,  {\partial}_{O}
  \partial_{Q}{f}) \label{A3fffsimplified}
\end{eqnarray*}

We collect terms with the same factor of $\partial_{q_i}\Pi^{IJ}$ from
derivatives of the bivector. Such a contribution with $\partial_{q_1}\Pi^{IJ}$
is of the form
\begin{eqnarray*}
&&\Pi^{NO}\left(-2{\partial}_{q_1}{{\Pi}^{UQ}}\,
{\partial}_{U}{f}\,{\partial}_{p_1}\partial_{N}{f}\,{\partial}_{O}
\partial_{Q}{f}+ {\partial}_{q_1}{{\Pi}^{UQ}}\,
{\partial}_{p_1}{f}\,  {\partial}_{N}\partial_{U}{f}\,
{\partial}_{O}\partial_{Q}{f}\right)\\
&=& \Pi^{NO}\bigl( \partial_{p_1}f \left(-\partial_{q_1}\Pi^{p_1 Q}
    {\partial}_{p_1}\partial_{N}{f} +
    {\partial}_{q_1}{{\Pi}^{p_2Q}}\,
 {\partial}_{N}\partial_{p_2}{f}+ {\partial}_{q_1}{{\Pi}^{p_3Q}}\,
  {\partial}_{N}\partial_{p_3}{f}\right)\\
&& -2\partial_{p_2}f {\partial}_{q_1}{{\Pi}^{p_2Q}}\,
{\partial}_{p_1}\partial_{N}{f}-
2\partial_{p_3}f {\partial}_{q_1}{{\Pi}^{p_3Q}}\,
{\partial}_{p_1}\partial_{N}{f}\bigr){\partial}_{O}
\partial_{Q}{f}\,,
\end{eqnarray*}
arranging by factors of first-order derivatives $\partial_{p_i}f$. By our
assumptions on $f$, the index $N$ is determined in all terms for non-zero
contributions and we obtain
\begin{eqnarray*}
 && \Bigl(\partial_{p_1}f \left(-\Pi^{p_1 O}\partial_{q_1}\Pi^{p_1 Q}
    {\partial}_{p_1}^2{f} + \Pi^{p_2 O}
    {\partial}_{q_1}{{\Pi}^{p_2Q}}\,
 \partial_{p_2}^2{f}+ \Pi^{p_3O}{\partial}_{q_1}{{\Pi}^{p_3Q}}\,
  \partial_{p_3}^2{f}\right)\\
&& -2\partial_{p_2}f \Pi^{p_1O}{\partial}_{q_1}{{\Pi}^{p_2Q}}\,
{\partial}_{p_1}^2{f}-
2\partial_{p_3}f \Pi^{p_1O}{\partial}_{q_1}{{\Pi}^{p_3Q}}\,
{\partial}_{p_1}^2{f}\bigr)\Bigr){\partial}_{O}
\partial_{Q}{f}\\
&=& \sum_O \left(\partial_{p_1}f \left(-\Pi^{p_1 O}\partial_{q_1}\Pi^{p_1 O}
    {\partial}_{p_1}^2{f} + \Pi^{p_2 O}
    {\partial}_{q_1}{{\Pi}^{p_2O}}\,
 \partial_{p_2}^2{f}+ \Pi^{p_3O}{\partial}_{q_1}{{\Pi}^{p_3O}}\,
  \partial_{p_3}^2{f}\right)\right.\\
&& \left.-2\partial_{p_2}f \Pi^{p_1O}{\partial}_{q_1}{{\Pi}^{p_2O}}\,
{\partial}_{p_1}^2{f}-
2\partial_{p_3}f \Pi^{p_1O}{\partial}_{q_1}{{\Pi}^{p_3O}}\,
{\partial}_{p_1}^2{f}\bigr)\partial_{O}^2f\right)
\end{eqnarray*}
setting $O=Q$ in the last step, again by our assumptions on $f$. We now go
through all remaining choices of the only free index $O$. All contributions to
terms containing $\partial_{q_1}\Pi^{p_1O}$ cancel out. We arrive at
\begin{eqnarray*}
 && 2\partial_{p_1}f
 \Pi^{p_2p_3}\partial_{q_1}\Pi^{p_2p_3} \partial_{p_2}^2f \partial_{p_3}^2f-
 2\partial_{p_2}f
 \Pi^{p_1p_3}\partial_{q_1}\Pi^{p_2p_3} \partial_{p_1}^2f\partial_{p_3}^2f  -
 2\partial_{p_3}f
 \Pi^{p_1p_2}\partial_{q_1}\Pi^{p_3p_2} \partial_{p_1}^2f\partial_{p_2}^2f\\
&=& 2\partial_{q_1}\Pi^{p_2p_3}
\sum_{\sigma\in
  Z_3}\Pi^{p_{\sigma(1)}p_{\sigma(2)}} \partial_{p_{\sigma(3)}}f
\partial_{p_{\sigma(1)}}^2f \partial_{p_{\sigma(2)}}^2f \,.
\end{eqnarray*}
Bringing back contributions with the remaining $\partial_{q_i}\Pi^{IJ}$, we
have (\ref{A3f}).
\end{proof}

For specific choices of $f$ obeying the condition stated in the Lemma, we can
compute $A_3(f,f,f)$ more explicitly. The first parenthesis in (\ref{A3f}) is
half the Jacobiator of the bivector, which is non-zero for a monopole star
product. The sum over cyclic permutations depends on the specific $f$.

\begin{ex} Let $\Pi$ be a bivector as stated in the conditions on a monopole
  star product.
\begin{enumerate}
\item Let $f=|p|^2=p_1^2+p_2^2+p_3^2$. We have
\[
 \sum_{\sigma\in
  Z_3}\Pi^{p_{\sigma(1)}p_{\sigma(2)}} \partial_{p_{\sigma(3)}}f
\partial_{p_{\sigma(1)}}^2f \partial_{p_{\sigma(2)}}^2f = 8\sum_{\sigma\in
  Z_3}\Pi^{p_{\sigma(1)}p_{\sigma(2)}} p_{\sigma(3)}\,.
\]
With a bivector as implied by (\ref{pp}),
\begin{equation}
 A_3(|p|^2,|p|^2,|p|^2) = \frac{32}{3}i (p\cdot B)\: {\rm div}B\,.
\end{equation}
For a monopole star product, ${\rm div}B\not=0$, and $p\cdot B$ is generically
non-zero for a charged particle with momentum $p$ moving in the magnetic field
$B$. Therefore, the a monopole star product obtained from a Weyl star product
cannot be alternative to third order in $\lambda$.
\item Another example in which (\ref{A3f}) can be used is $f=e^{i\alpha_1
    p_1}+e^{i\alpha_2 p_2}+e^{i\alpha^3p_3}$ for
  $(\alpha_1,\alpha_2,\alpha_3)\in{\mathbb R}^3$, a family of bounded
  functions. The sum over cyclic permutations then equals
\[
\sum_{\sigma\in
  Z_3}\Pi^{p_{\sigma(1)}p_{\sigma(2)}} \partial_{p_{\sigma(3)}}f
\partial_{p_{\sigma(1)}}^2f \partial_{p_{\sigma(2)}}^2f=
i\alpha_1^2\alpha_2^2\alpha_3^2 \left(\frac{\Pi^{p_1p_2}}{\alpha_3}+
  \frac{\Pi^{p_2p_3}}{\alpha_1}+ \frac{\Pi^{p_3p_1}}{\alpha_2}\right)
e^{i(p_1+p_2+p_3)}\,.
\]
For a bivector as in (\ref{pp}), we have
\begin{eqnarray} \label{A3bounded}
&& A_3(e^{ip_1}+e^{ip_2}+e^{ip_3},
    e^{ip_1}+e^{ip_2}+e^{ip_3}, e^{ip_1}+e^{ip_2}+e^{ip_3})\\
& =&    -\frac{4}{3}\alpha_1^2\alpha_2^2\alpha_3^2
    e^{i(\alpha_1p_1+\alpha_2p_2+\alpha_3p_3)}
    \left(\frac{B^1}{\alpha_1}+\frac{B^2}{\alpha_2}+\frac{B^3}{\alpha_3}\right)
{\rm div}B\,. \nonumber
\end{eqnarray}
For any non-zero $B$, there is a triple $(\alpha_1,\alpha_2,\alpha_3)$ such
that $B^1/\alpha_1+B^2/\alpha_2+B^3/\alpha_3$ is not identically
zero. Therefore, every magnetic field with non-zero divergence gives rise to
an $f$ with $A_3(f,f,f)\not=0$.
\end{enumerate}
\end{ex}

The Lemma implies non-alternativity of monopole star products obtained from a
Weyl star product quantizing (\ref{Pi}), but this already follows from
Theorem~\ref{Theorem}. Having explicit examples with $A_3(f,f,f)\not=0$
implies further results.

A property weaker than alternativity is \emph{flexibility}, for which, by
definition, only anti-symmetry with respect to the first and third entry is
required:
\begin{equation}
 A(f,g,h)=-A(h,g,f) \, .
\end{equation}
Flexibility is important for quantum mechanics because it is a necessary and
sufficient condition \cite{Okubo} for the commutator
\begin{equation}
 [f,g]=f\star g-g\star f
\end{equation}
to be a derivation of the Jordan product
\begin{equation}
 f\circ g:=\frac{1}{2} (f\star g-g\star f)\,.
\end{equation}
Heisenberg equations of motion
\begin{equation}
 \frac{{\rm d}f}{{\rm d}t} = \frac{[f,H]}{i\hbar}
\end{equation}
with a Hamiltonian $H$ then obey a product rule of the form
\begin{equation}
 \frac{{\rm d}(f\circ g)}{{\rm d}t} = \frac{{\rm d}f}{{\rm d}t}\circ g+ f\circ
 \frac{{\rm d}g}{{\rm d}t}\,.
\end{equation}
To second order in $\lambda$, flexibility of the associator follows from
(\ref{A2B}) for any star product with symmetric $B_2$.  However, as with
alternativity, this fact does not guarantee that flexibility is realized at
higher orders.

Another condition weaker than alternativity is {\em power-associativity:} A
power-associative algebra is defined as an algebra ${\cal A}$ such that the
subalgebra generated by any single element $a\in{\cal A}$ is associative. For
any positive integer $n$, the $n$-th power $a^n$ is then uniquely defined even
though the algebra product may be non-associative. For Weyl star products of
monopole systems, we have

\begin{theo}
  A Weyl star product which quantizes (\ref{Pi}) with ${\rm div}B\not=0$
  cannot be flexible or power associative.
\end{theo}
\begin{proof}
  Since there is an $f$ such that $A_3(f,f,f)\not=0$, the associator cannot be
  antisymmetric in its first and last arguments. Moreover, we have $f\star
  (f\star f)-(f\star f)\star f=A_3(f,f,f)\lambda^3+O(\lambda^4)$ and the
  subalgebra generated by $f$ cannot be associative.
\end{proof}

\section{Conclusions}

We have shown that, under rather weak conditions, star products that quantize
the phase space of a charged particle in the presence of a magnetic monopole
density cannot be alternative. More generally, we have provided obstructions
for a non-associative star product with symmetric $B_2$ being alternative. By
the non-associative Gelfand--Naimark theorem \cite{NonAssNormed}, this result,
together with the fact that the algebra is unital, implies that there is no
norm that would turn the quantum algebra into a $C^*$-algebra, even if the
algebra can be restricted to bounded functions; see (\ref{A3bounded}). This
version of our result strengthens the usual statement that non-associative
systems cannot be quantized in the standard way by representing observables on
a Hilbert space. One way to circumvent the use of Hilbert spaces in
associative systems is to take an algebraic view point and define quantum
states as positive linear functionals on the $C^*$-algebra of bounded
observables; see for instance \cite{LocalQuant}. For non-associative systems
of the kind studied here, this route must be generalized because the
star-product algebra cannot be turned into a $C^*$-algebra. One can still use
positive linear functionals, but only on a $*$-algebra.

Non-alternativity rules out the use of octonions as realizations of
observable algebras of the relevant physical systems. Recently, in
\cite{NonAssOcto}, octonions have been used to realize the relations
(\ref{pp}) and (\ref{Asso}) for linear functions of the momentum
components. An extension to non-linear functions would encounter the same
obstructions found here for star products, and a purely octonionic
construction would no longer suffice.

\section*{Acknowledgements}
We thank Peter Schupp and Stefan Waldmann for useful discussions.  This work
was supported in part by NSF grants PHY-1307408 and PHY-1607414. The work of
T.S.~was additionally supported in part by Projeto P.V.E. 88881.030367/
2013-01 (CAPES/Brazil), the project MODFLAT of the European Research Council
(ERC), and the NCCR SwissMAP of the Swiss National Science Foundation.

\begin{appendix}
\section{Details on the derivation of eq.~(\ref{A3final})}
\label{a:A3}

Starting from (\ref{A3}), and using all its cyclic permutations, we can write
the fully anti-symmetric part of $A_3$ as
\begin{eqnarray}
6A_3(f,g,h)^-&=&B_2(f, B_1(g,h)) -B_2(B_1(f,g), h) +B_2(h,
B_1(f,g)) \\
& &-B_2(B_1(h,f) g)+B_2(g, B_1(h,f)) -B_2(B_1(g,h), f)\nonumber\\
& &-B_2(f, B_1(h,g)) +B_2(B_1(f,h), g) - B_2(g, B_1(f,h))\nonumber\\
& &+B_2(B_1(g,f), h) -B_2(h, B_1(g,f)) +B_2(B_1(h,g), f) +(B_1 \leftrightarrow
B_2)\,.\nonumber
\end{eqnarray}
Using the definition of the anti-symmetric parts of the $B_i$, we have
\begin{eqnarray}
6A_3(f,g,h)^-&=&2B_2^-(f, B_1(g,h)) +2B_2^-(h, B_1(f,g))
+2B_2^-(g, B_1(h,f))\\
& &-2B_2^-(f, B_1(h,g)) -2B_2(g, B_1(f,h)) -2B_2^-(h, B_1(g,f)) +(B_1
\leftrightarrow B_2)\,.\nonumber
\end{eqnarray}
Finally, using the fact that the $B_i$ are linear in their arguments, we
obtain the required form for the fully anti-symmetric part of $A_3$ as in
(\ref{A3final}).

\section{Proof of Lemma \ref{lemma3} without Pentagon identity}
\label{a:Pentagon}
To begin with, let us write the third-order associator as before:
\begin{eqnarray} \label{A32}
A_3(f, g, h) &=& {\rm d}B_3(f, g, h) + B_2(f, B_1(g,h))\nonumber\\ &
&\hspace{3mm}-B_2(B_1(f,g), h) + B_1(f, B_2(g,h)) - B_1(B_2(f,g), h)\,,
\end{eqnarray}
where ${\rm d}B_n=f B_n(g, h) +B_n(f, gh) -hB_n(f,g) - B_n(fg, h)$.If we apply
the Hochschild coboundary operator to $A_3$, the first term in (\ref{A32})
should give zero because ${\rm d}^2=0$. (Again, when applied to coefficients
in an $\lambda$-expansion of a non-assocative star product, only the associative
multiplication of smooth functions is used in the definition of ${\rm d}$.)
However, for completeness we will explicitly show this. The part in ${\rm
  d}A_3(f,g,h,k)$ involving contributions only from the $B_3$ terms has the
form
\begin{eqnarray}
 f\,{\rm d}B_3(g,h,k) -{\rm d}B_3(fg,h,k)+{\rm d}B_3(f,gh,k)-{\rm
   d}B_3(f,g,hk)+k\,{\rm d}B_3(f,g,h)\,.
\end{eqnarray}
Using the definition of ${\rm d}B_n$ for $n=3$ gives
\begin{eqnarray}
 &&f\bigg(g\,B_3(h,k)+B_3(g,hk)-k\,B_3(g,h)-B_3(gh,k)\bigg)\nonumber\\
&&-\bigg(f\,g\,B_3(h,k)+B_3(fg,hk)-k\,B_3(fg,h)-B_3(fgh,k) \bigg)\nonumber \\
&&+\bigg(f\,B_3(gh,k)+B_3(f,ghk)-k\,B_3(f,gh)-B_3(fgh,k)\bigg)\nonumber\\
&&-\bigg(f\,B_3(g,hk)+B_3(f,ghk)-h\,k\,B_3(f,g)-B_3(fg,hk)\bigg)\nonumber\\
&&+k\bigg(f\,B_3(g,h)+B_3(f,gh)-h\,B_3(f,g)-B_3(fg,h)\bigg)\,.\nonumber
\end{eqnarray}
Upon a close inspection of this expression, we see that there is a counterterm
for each term, and thus it is zero. We are left with the action of the
coboundary operator on the last four terms in (\ref{A32}). Concentrating, for
now, on its action on the $B_2$ terms, using the generic definition of ${\rm
  d}B_n$ for $n=2$, we obtain a part in ${\rm d}A_3(f,g,h,k)$ that is of the
form:
\begin{eqnarray}
 &&-f\bigg(\,B_2(g,B_1(h,k))-\,B_2(B_1(g,h),k)\bigg)\nonumber \\
&&-B_2(fg,B_1(h,k))+B_2(B_1(fg,h),k) \nonumber \\
&& +B_2(f,B_1(gh,k))-B_2(B_1(f,gh),k)\nonumber \\
&&-B_2(f,B_1(g,hk))+B_2(B_1(f,g),hk) \nonumber \\
&&+k\bigg( B_2(f,B_1(g,h))-B_2(B_1(f,g),h)\bigg)\,. \label{A33}
\end{eqnarray}
Using the Leibniz property of $B_1$, and removing terms that identically
cancel out, we are left with
\begin{eqnarray}
 &&-f\,B_2(g,B_1(h,k))-f\,B_2(B_1(g,h),k)-B_2(fg,B_1(h,k))\nonumber\\
&&+B_2(fB_1(g,h),k)+B_2(f,g
 B_1(h,k))-B_2(hB_1(f,g),k)\nonumber \\
&&-B_2(f,kB_1(g,h))+B_2(B_1(f,g),hk)+k\,B_2(f,B_1(g,h))-
k\,B_2(B_1(f,g),h)\,.\nonumber
\end{eqnarray}
This expression can be cast into a more succinct form in terms of ${\rm
  d}A_2$, by adding and subracting a few terms as follows:
\begin{eqnarray} \label{A34}
&&{\rm d}B_2(f,g,B_1(h,k))-{\rm d}B_2(f,B_1(g,h),k)+{\rm
  d}B_2(B_1(f,g),h,k)\\
&&+B_1(h,k)B_2(f,g)-B_2(h,k)B_1(f,g) \,. \nonumber
\end{eqnarray}
The action of the differential on the $B_1$ terms in (\ref{A32}) gives an
expression similar to (\ref{A33}), with the roles of $B_1$ and $B_2$
exchanged. Again upon using the Leibniz property of $B_1$ and cancelling
terms, we have the contribution to ${\rm d}A_3$ as
\begin{eqnarray}
&&-f\,B_1(B_2(g,h),k)-g\,B_1(f,B_2(h,k))+B_1(B_2(fg,h),k)+B_1(f,B_2(gh,k))
\nonumber
\\
&&-B_1(B_2(f,gh),k)-B_1(f,B_2(g,hk))+h\,B_1(B_2(f,g),k)+k\,B_1(f,B_2(g,h))\,.
\nonumber
\end{eqnarray}
Using anti-symmetry and linearity in either of the arguments of $B_1$, and
again adding and subtracting a few terms, we introduce
${\rm d}B_2$ as
\begin{eqnarray} \label{A35}
 B_1({\rm d}B_2(g,h,k),f)-B_1({\rm
   d}B_2(f,g,h),k)-B_2(f,g)B_1(h,k)+B_2(h,k)B_1(f,g) \,.
\end{eqnarray}
As the final result, (\ref{A34}) and (\ref{A35}) give
\begin{eqnarray}
 {\rm d}A_3(f,g,h,k) &=& {\rm d}B_2(f,g,B_1(h,k))-{\rm
   d}B_2(f,B_1(g,h),k)+{\rm d}B_2(B_1(f,g),h,k)\nonumber \\
&&+B_1({\rm d}B_2(g,h,k),f)-B_1({\rm d}B_2(f,g,h),k)\,.
\end{eqnarray}
To get the same result as in (\ref{dA3}), which was obtained using the
Pentagon identity, we just use the definition of ${\rm d}B_2$ in terms of the
second-order associator as ${\rm d}B_2(f,g,h) = A_2(f,g,h)
-B_1(f,B_1(g,h))+B_1(B_1(f,g),h)$, and use the linearity of $B_1$ in its first
argument in the last two terms.

\end{appendix}


\newcommand{\noopsort}[1]{}

\end{document}